\documentclass{llncs}
\usepackage{amsmath}
\usepackage{graphicx}
\usepackage{txfonts}
\usepackage{hyperref}
\usepackage{algorithm}
\usepackage{algorithmic}
\newenvironment{keywords}{
       \list{}{\advance\topsep by0.35cm\relax\small
       \leftmargin=1cm
       \labelwidth=0.35cm
       \listparindent=0.35cm
       \itemindent\listparindent
       \rightmargin\leftmargin}\item[\hskip\labelsep
                                     \bfseries Keywords:]}
     {\endlist}
\begin{document}
\title{Balls into Bins with Strict Capacities and Weighted Edges}
\author{Ankur Sahai}
\maketitle
\pagestyle{plain}
\begin{abstract}        
We explore a novel theoretical model for studying the performance of distributed storage management systems where the data-centers have limited capacities (as compared to storage space requested by the users). Prior schemes such as Balls-into-bins (used for load balancing) neither consider bin (consumer) capacities (multiple balls into a bin) nor the future performance of the system after, balls (producer requests) are allocated to bins and restrict number of balls as a function of the number of bins. Our problem consists of finding an optimal assignment of the online producer requests to consumers (via weighted edges) in a complete bipartite graph while ensuring that the total size of request assigned on a consumer is limited by its capacity. The metric used to measure the performance in this model is the (minimization of) weighted sum of the requests assigned on the edges (loads) and their corresponding weights. We first explore the optimal offline algorithms followed by the analysis of different online techniques (by comparing their performance against the optimal offline solution). LP and Primal-Dual algorithms are used for calculating the optimal offline solution in $O(r \cdot n)$ time (where r and n are the number of requests and consumers respectively) while randomized algorithms are used for the online case.\\

We propose randomized online algorithms in which the consumers are selected based on edge probabilities (that can change with consumer failures; due to capacity exhaustion) and evaluate the performance of these randomized schemes using probabilistic analysis. The performance of the online algorithms is measured using competitive analysis assuming an oblivious adversary who knows the randomized algorithm but not the results produced. For the simplified model with equal consumer capacities an average-case competitive ratio (which compares the average cost of the output produced by the online algorithm and the minimum cost of the optimal offline solution) of $\big( \frac{\overline{d}}{\min_{i, j} d_{i,j}} \big)$ (where d is the edge weight / distance) is achieved using an algorithm that has equal probability for selecting any of the available edges with a running time of $O(r)$. In the extending the model to arbitrary consumer capacities we show an average case competitive ratio of $\big( \frac{\overline{d \cdot c}}{\overline{c} \cdot \min_{i,j} d_{i,j}} \big)$. This theoretical model gives insights to a (storage) cloud system designer about, how the different attributes (producer requests, edge weights and consumer capacities) effect the overall (read / write) performance of a distributed storage management system over a period of time.

\end{abstract}
%
\section{Motivation}
\label{sec:motivation}
A (storage cloud) is a collection of VMs (represented by consumers in our model in section \ref{sec:problem-definition}) and data-centers (producers) that interact with each other through network links with different bandwidths / data-transfer rates (inversely proportional to the edge weights). VMs generate requests for virtual disk space (requests) at the time of creation. The data-centers have limited storage space (capacity). One of the core problems is to optimally select a data-center for allocating the disk space for a VM. If we assume that all data-centers are available for the requests generated by each VM then this configuration forms a complete bipartite graph. As the VM users are added gradually to the cloud and their disk-space requirements need to be satisfied instantly to meet the SLAs this forms the online part of the problem. We assumes that once a VM is allocated storage space on a particular data-center it cannot be moved to a different one, thus VMs read / write data using the same network links (fixed at the time of VM provisioning) throughout their lifetime.

The overall I/O performance of the cloud storage system can be optimized when the quality of network links used for the majority of read / write operations is maximized which, can be acheived by using the higher quality (inversely proportional to edge weights) links for allocating as much producer requests as possible. This is equivalent to minimizing the the distance weighted sum of request allocated in the system. This theoretical model can be used to measure the performace of distributed storage provisionig schemes such as VMware's virtualization framework - VSphere \cite{vmware-drs,vmware-scale-storage}.

\section{Problem Definition}
\label{sec:problem-definition}

In the offline model ${\cal M_{OFF}}$ the requests sizes $s(t)$, capacities $c_j$ and edge weights $d_{i,j}$ are arbitrary.

\begin{itemize}
\item A complete (undirected) bipartite graph $G = (P, C, E)$.
\item Set $P = \{P_i \: | \: i \in [m]\}$ of producers, $m = |P|$. Set $C = \{C_j \: | \: j \in [n]\}$ of consumers, $n = |C|$, with arbitrary capacity~$c$. The capacity is a strict bound for the consumer's load and cannot be exceeded.
\item Set of edges $E = \{e_{i,j} \bigm\vert i \in [m], j \in [n]\}$ connect each producer to each consumer and with $d_{i,j}$ which denote the distance between producer $P_i$ and $c_j$.
\item $r$ is the total number of requests. The requests are produced by the producers and allocated to the consumers. A request can be split and allocated to different producers. s(t) is the size of the request produced at time t.
\item Producers are selected at random in each round.
\item $d(t)$ is the distance between the producer and consumer that are chosen in time step $t$. The total cost of the random allocation is $G = \sum_{t=1}^{r} s(t) \cdot d(t)$.
\item We assume that every request can commit to at least one consumer (regardless of the previous allocation), i.~e. at least one consumer must have sufficient space left. The number and size of requests must be limited accordingly.
\item $\ell_j$ is the load on consumer $c_j$. All loads $\ell_1, ..., \ell_n$ are set to $0$ in the beginning. Whenever a consumer receives a request, its load is increased by the size of the request. 
\item $\ell_{i, j}$ is the total request from producer $P_i$ allocated to consumer $C_j$ known as the edge load. All loads $\ell_{1,1}, ..., \ell_{m,n}$ are set to $0$ in the beginning. Whenever a consumer receives a request from a producer the corresponding edge load is increased by the size of the request.

\end{itemize}

The objective of this problem is to find an average-case \cite{average-case} $\alpha$-competitive online algorithm (assuming startup cost $\beta = 0$), that minimizes the expected value of weighted sum of edge loads while satisfying the producer requests:

\begin{equation}
\sum_{i \in [m], j \in [n]} E_I\bigg[ d_{i,j}(t)  \cdot  l_{i,j}(t) \bigg] \le \alpha \cdot  OPT_I(t), \forall (t, I)
\label{operation 5}
\end{equation} 

where, $\alpha$ is a constant and $OPT_I(t)$ is the output of the optimal offline algorithm for the input received in the time interval $(0, t]$ for an instance I.

Equation (\ref{operation 6}) guarantees that the total load on the edges \big( $\sum_{j \in [n]} l_{i,j}(t)$ \big) incident on each producer ($i \in [m]$) is equal to the total size of the requests generated by the producer \big($\sum_{t} R_i(t)$\big). This will be referred to as \emph{producer request} constraint.

\begin{equation}
\sum_{j \in [n]} l_{i,j}(t) = \sum_{t} R_i(t), \forall (i \in [m], \; t)
\label{operation 6}
\end{equation}

Equation (\ref{operation 7}) ensures that the total load on the edges \big($\sum_{i \in [m]} l_{i,j}(t)$\big) incident on a consumer ($j \in [n]$) does not exceed the consumer capacity ($c_{j}$). This will be referred to as \emph{consumer capacity} constraint.

\begin{equation}
\sum_{i \in [m]} l_{i,j}(t) \le c_{j}, \forall (j \in [n], \; t)
\label{operation 7}
\end{equation}

\section{Related Work}
\emph{Balls-into-Bins} \cite{balls-to-bins-tight,klaus-thesis,balanced-alloc} model is used for studying load balancing in a similar resource allocation configuration where, the objective is to place m balls into n bins while guaranteeing bounds on the maximum, minimum or the average load across all the bins. The main advantage of the model defined in section \ref{sec:problem-definition} is that it also takes into accounts the capacities of the bins (which are analogous to the consumers in the problem defined in section \ref{sec:problem-definition}). Further the model described in section \ref{sec:problem-definition} compares the performance of a randomized algorithm with the optimal offline algorithm. 

One of the well studied online algorithm for assigning resources to users is the \emph{k}-server problem \cite{k-server-randomized} where, the servers handle request once for each client. Dynamic assignment \cite{dynamic-hungarian-algo} has a similar configuration involving bipartite graphs.

\section{Offline Algorithms}
\label{sec:offline}
The optimal offline algorithm has to exhaustively look at the available edges for allocating a request. This paper uses Linear Programming and Primal-Dual algorithms for solving the offline version (section \ref{sec:problem-definition}).

\subsection{Linear Programming}
\label{sec:offline-LP}
Linear Programming \cite{lp-onlline} is a method used to solve large-scale optimization problems with a set of constraints and an objective function (minimization or maximization) both being linear.

The LP formulation for this problem is as follows,

Objective function:

Minimize:
\begin{equation}
\sum_{i \in [m], j \in [n]} d_{i,j}(t)  \cdot  l_{i,j}(t), \; \; \; d_{i,j}(t) \ge 0, \; l_{i,j}(t) \ge 0
\label{operation 8}
\end{equation}

Constraints:
\begin{equation}
\sum_{j \in [n]} l_{i,j}(t) \ge \sum_{t} R_i(t), \;\; R_i(t) >0, \;\;  \forall (i \in [m], \; t)
\label{operation 9}
\end{equation}

\begin{equation}
\sum_{i \in [m]} l_{i,j}(t) \le c_{j} \implies -\sum_{i \in [m]} l_{i,j}(t) \ge - c_{j}, \; \forall (j \in [n]\; t)
\label{operation 10}
\end{equation}

This LP is used for calculating the optimal solution OPT(t) for the input received in the time interval $(0, t], t$. As the requests $R_i, i \in [m]$ are non-negative the load assignments $l_{i,j}$ in the objective function (\ref{operation 8}) are also non-negative. Equation (\ref{operation 9}) represents the \emph{producer request} constraint corresponding to (\ref{operation 6}) whereas (\ref{operation 10}) corresponds to the \emph{consumer capacity} constraint corresponding to (\ref{operation 7}). In addition to this, new constraints corresponding to the existing load assignment on edges have to be added at each time instance $t$ for \emph{assignment without reallocation} (defined in section \ref{sec:problem-definition}).

LP formulation in section \ref{sec:offline-LP} produces a feasible \big( without violating the consumer capacity constraints (\ref{operation 7}) \big) assignment of loads $l_{i,j}$ on edges $e_{i,j}$ corresponding to the requests $R_i(t), \; \forall (i \in [m], t)$. Equation (\ref{operation 9}) guarantees that the total request generated by producers $i \in [m]$ is satisfied and (\ref{operation 10}) ensures that the capacities of consumers $j \in [n]$ are not exceeded. By definition of problem in section \ref{sec:problem-definition}, this is a valid assignment of loads on edges. LP formulation in section \ref{sec:offline-LP} produces the optimal assignment of loads $l_{i,j}$ on edges $e_{i,j}$ corresponding to the requests $R_i(t), \; \forall (i \in [m], t)$. The solution produced by LP is optimal as fractional loads are allowed and the LP considers all possible solutions minimize the objective function (\ref{operation 8}). 

\subsection{Primal-Dual}
\label{sec:offline-PD}
Primal-Dual algorithms are used for a certain class of optimization problems where there are a finite number of feasible solutions available at each step. The dual is often useful for providing intuitions about the nature of the solution that are implicit in the primal.

Consider the dual of the LP formulation in section \ref{sec:offline-LP}. Let $y_i$ be the dual variables corresponding to producers $i \in [m]$ (\ref{operation 9}) and $z_j$ be the dual variables corresponding to the consumers $j \in [n]$ (\ref{operation 10}) then the corresponding dual is,

Objective function:

Maximize:

\begin{equation}
\sum_{i \in [m]} y_i  \cdot  R_i - \sum_{j \in [n]} z_j  \cdot  c_j, \; y_i \ge 0, \; z_j \ge 0
\label{operation 11}
\end{equation}

$R_i = \sum_{t=1}^{r} R_i(t)$ is the total size of the requests produced by $P_i$ 

Constraints:
\begin{equation}
\displaystyle y_i - z_j \le d_{i,j} , \forall (i \in [m], j \in [n])
\label{operation 12}
\end{equation}

Equation (\ref{operation 12}) suggests that the potential difference between producers and consumers can be atmost equal to $d_{i,j}$. This will be referred to as dual \emph{potential limit} constraint.

According to complementary slackness conditions,
\begin{equation}
\displaystyle l_{i,j} > 0 \iff y_i - z_j = d_{i,j}
\label{operation 13}
\end{equation}

Let T be the set of tight constraints (which represent edges selected by Algorithm \ref{alg1}) \big($T = \{(i,j) \; | \; y_i - z_j = d_{i,j}, \forall (i \in [m], j \in [n]) \}$\big) and {\bf S}, the set of slack constraints \big(${\bf S} = \{ (i,j) \; | \; y_i - z_j < d_{i,j}, \forall (i \in [m], j \in [n]) \}$\big). Then consider an \emph{unit benefit} function which measures the increase in dual objective (\ref{operation 11}) corresponding to a unit amount of request \big($\Delta y_i \cdot R_i(t)$\big) produced by producer $i \in [m]$ at time $t$ where, $\Delta z_j$ are the corresponding increases in the consumer-dual variables; required for keeping the dual \emph{potential limit} constraints (\ref{operation 12}) tight.

\begin{align}
B_i(t) &= \Delta y_i  \cdot R_i(t) - \sum_{j \in T}   \Delta z_j  \cdot c_j
\label{operation 17}
\end{align}

The \emph{unit benefit} function in \ref{operation 17} is used as a criteria for selecting the request which produces the maximum increase in the dual objective function by the Primal-dual algorithm.

\begin{algorithm}[H]
\caption{Primal-Dual algorithm for \emph{Offline Assignment}}
\label{alg1}
\begin{algorithmic}[1]
\STATE $T \gets \emptyset$
\STATE ${\bf S} \gets \{ (i,j) \; | \; \forall (i \in [m], j \in [n]) \}$
\STATE $y_{i} \gets 0, \forall i \in [m]$
\STATE $z_j \gets 0, \forall j \in [n]$
\STATE $l_{i,j} \gets 0, \forall (i \in [m], j \in [n])$
\WHILE{ $\exists (i, t): B_i(t) \ge 0$}
\STATE $y_i: \max_i \big( B_i(t) \bigg)$
\STATE $z_j: \min_{j} \{d_{i,j} - (y_i - z_j) \;| \; i \in [m] ,j \in [n] \}$
\STATE $\Delta_1 = d_{i,j} - (y_i - z_j)$
\STATE $y_i \gets y_i + \Delta_1$
\STATE $z_j \leftarrow z_j + \Delta_1, \forall j \in T$
\STATE $T = T \cup (i,j)$
\STATE ${\bf S} = {\bf S} \setminus (i,j)$
\STATE $\Delta_2 = c_j - \sum_{(i, j) \in T} l_{i,j}$
\STATE $l_{i,j} \gets l_{i,j} + \Delta_2$
\ENDWHILE
\end{algorithmic}
\end{algorithm}

 The primal-dual Algorithm \ref{alg1} chooses the (producer dual variable $y_i$ corresponding to) request $R_i(t)$ with the highest benefit function $B_i(t)$ (step 7 in Algorithm \ref{alg1}) to maximize the increase in value of dual objective function and then chooses the dual \emph{potential limit} constraint ({operation 12}) (and the corresponding variable $z_j$) that is closest to becoming tight at step 8 in Algorithm \ref{alg1}  (corresponding to the least cost edge in the primal) and increases the value of $y_i$ by the amount ($\Delta_1$, in steps 9 and 10 in Algorithm \ref{alg1}) that is needed to make this constraint tight (which corresponds to selecting an edge for allocating a request in the primal using complementary slackness condition \ref{operation 13}). If this constraint is tight (T updated in step 12) then corresponding $z_j$ are also increased by $\Delta_1$ (step 11 in Algorithm \ref{alg1}) to maintain tightness. 

Increasing $d_{i,j}$ by amount $\Delta_2$ is only symbolic / superficial. The practical output of the algorithm can be traced by looking up the set of dual constraints (corresponding to edges in primal) that are selected corresponding to each request. Intuitively, optimal offline primal-dual Algorithm \ref{alg1} assigns the lower sized requests to the higher cost edges and uses the lower cost edges for the higher sized requests.

\begin{theorem}  [Optimality of Algorithm \ref{alg1}] The Primal-Dual Algorithm \ref{alg1} reaches the optimal solution for the \emph{assignment without reallocation} problem in section \ref{sec:problem-definition} when it is not possible to increase the cost of the dual objective function (\ref{operation 11}).
\end{theorem}
\begin{proof}
nitializing $y_i \gets 0, \forall i \in [m]$ and $z_j \gets 0, \forall j \in [n]$ (steps 1 and 2 in Algorithm \ref{alg1}) and from the way we increase the dual variables (step 9 in Algorithm \ref{alg1}) the dual \emph{potential limit constraint} (\ref{operation 12}) is always satisfied. Setting $l_{i,j} \gets 0, \forall (i \in [m], j \in [n])$ (step 5 in Algorithm \ref{alg1}) makes sure that the primal is feasible at the start. The primal \emph{consumer capacity} constraints in (\ref{operation 10}) are not violated by the way we increase $l_{i,j}$ from steps 14 and 15 in Algorithm \ref{alg1}.

When the cost of the dual objective function (\ref{operation 11}) cannot be increased further,

\begin{align}
\sum_{i \in T} \Delta y_i  \cdot  R_i - \sum_{j \in T} \Delta z_j  \cdot  c_j &= 0\\
\implies \sum_{i \in T} \Delta y_i  \cdot  R_i &= \sum_{j \in T} \Delta z_j  \cdot  c_j
\label{operation 201}
\end {align}

At this point, if there was a pending request $R_i(t)$ with a positive \emph{unit benefit} function that could choose one of the slack (S) dual constraints \big( without violating any of the \emph{potential limit} dual constraints(\ref{operation 12}) \big) then the corresponding producer dual variable $y_i$ can be increased until the dual constraint became tight. By contradiction such a request does not exists by definition (step 6) of Algorithm \ref{alg1} otherwise, the algorithm would have continued. So the system is in equilibrium and changing the value of any of the dual variables will violate (one or more) dual \emph{potential limit} constraints.

However, for the current solution \big($(y_i, z_j) \in T$\big), the dual variables need to be changed by an equal amount \big($\Delta y_i = \Delta z_j = \Delta, \; \forall (i \in [m], j \in [n])$\big) for maintaining the potential equilibrium [keeping the dual constraints in T tight for meeting complementary slackness (\ref{operation 13}) condition].

\begin{equation}
\Delta y_i = \Delta z_j = \Delta
\label{operation 200}
\end{equation}

When it is not possible to increase the value of dual objective function then using (12) and (13),

\begin{equation}
\sum_{i \in T}  R_i = \sum_{j \in T}  c_j
\label{operation 19}
\end{equation}

From steps 14 and 15 in Algorithm \ref{alg1}, we know that,
\begin{equation}
\sum_{j \in T} c_j = \sum_{(i, j) \in T} l_{i,j}(t)
\label{operation 20}
\end{equation}

Using (\ref{operation 19}) and (\ref{operation 20}) we infer that the requests have been met,
\begin{equation}
\sum_{i \in T}  R_i = \sum_{j \in T} l_{i,j}(t)
\label{operation 37}
\end{equation}

This means that the primal \emph{producer request} constraints in (\ref{operation 9}) are feasible and the primal \emph{consumer capacity constraints} in (\ref{operation 10}) and dual \emph{potential limit} constraints are always feasible. Complementary slackness (\ref{operation 13}) is satisfied as we only increase $d_{i,j}$ when the dual constraint is tight. This means that the primal is optimal and the dual is optimal. Thus, this is the optimal solution for the problem in section \ref{sec:problem-definition}.
\end{proof}

\begin{theorem} [Time complexity of Algorithm \ref{alg1}] The Primal-Dual Algorithm \ref{alg1} takes $O(r \cdot n)$, time to complete where, d and n the number of requests and consumers respectively.
\end{theorem}

\begin{proof} 
At each step of the Algorithm \ref{alg1} $O(r)$ (where r is the number of requests) operations are needed to select the (producer dual variable $y_i$ corresponding to the) request with maximum \emph{unit benefit} function, another $O(n)$ operations to find the corresponding consumer dual variable $z_j$. Hence, it takes $O (d + n)$ operations for the first (while loop from steps 6 to 16 in Algorithm \ref {alg1}) and $O (d-1 + n)$ operations for the second request by keeping track of the requests that have been covered by setting a flag in a hash table. This gives a time complexity of $O(r \cdot n)$ using amortized analysis. 
\end{proof}

\section{Online Algorithms}
\label{sec:online-algorithms}
In the beginning we look at the following simplified model ${\cal M}$,

\begin{itemize}
\item Let $R(t)$ be the request in round $t$ and let $C(t)$ denote the set of consumers in round $t$ that have enough space left to store request $R(t)$. 
\item All requests $R(t)$ have the same size s.
\item All consumers $C_i$ have equal capacities $c$.
\item An oblivious online adversary who can only manipulate the size of the online requests.
\end{itemize}

We use randomized algorithms for the online version as for a deterministic algorithm the oblivious online adversary who knows the edge selected in each round can manipulate the online request sizes to distort the value of the objective function (\ref{operation 5}).

The simplified algorithm ${\cal A}$ is:

\begin{algorithm}[H]
\caption{Online Algorithm ${\cal A}$}
\label{alg2}
\begin{algorithmic}[1]
\STATE Given: $P$; $C$; $c$; $d_{i,j}$; $s$
\STATE Initialize: $\ell_j = 0$ $\forall j \in [n]$; $G = 0$; $n = |C|$
\FOR {$t = 1$ to $r$}
\STATE Choose $P_i \in [m]$ independently and uniformly at random
\STATE $C(t) = \{c_j \: | \: j \in [n] \land \ell_j + s \leq c\}$
\STATE Choose $c_j \in [n]$ independently and uniformly at random  from $C(t)$.
\STATE $\ell_j = \ell_j + s$
\STATE $G = G + s \cdot d_{i,j}$
\ENDFOR
\end{algorithmic}
\end{algorithm}

\begin{theorem}  
Assuming model ${\cal M}$ and running algorithm ${\cal A}$, the expected total cost is $E[G] = \sum_{t=1}^{r} s \cdot \overline{d}$.
\end{theorem}
\begin{proof} Using the law of total probability over the different combination of consumer failures $|C_j| = n -k$ when there are $n-k$ consumers available,

\begin{align}
p\Big(e_{i=P_i, j} \bigm\vert (|C(t)| = n-k)\Big)   &= \sum_{j=1}^{{n \choose k}} p\Big(e_{i=P_i, j}  \bigm\vert |C_j = n -k\Big) \cdot p\Big(|C_j| = n -k\Big) \\
&= \frac{1}{{n \choose k}} \Bigg( \frac{{n \choose k} \cdot  \sum_{j=1}^{n} p(e_{i,j}) - \frac{{n \choose k} \cdot k}{n} \cdot \sum_{j=1}^{n} p(e_{i,j}) }{n-1} \Bigg) \\
&= \frac{\sum_{j=1}^{n} d_{i,j}}{n}
\label{operation 900}
\end{align}

Using (\ref{operation 900}) and the law of total expectation over the number of consumers available,

\begin{align}
E\Big[d_{i=P_i,j} \bigm\vert P_{i}\Big]  &= \sum_{k=1}^{n} E\Big[d_{i=P_i, j}  \bigm\vert (|C(t)| = n-k)\Big] \cdot p\Big( |C(t)| = n-k\Big) \\
&= \frac{\sum_{j=1}^{n} d_{i,j}}{n} \cdot \sum_{k=1}^{n} p\Big(|C(t)| = n-k\Big) \\
&= \frac{\sum_{j=1}^{n} d_{i,j}}{n}
\label{operation 115}
\end{align}

We calculate the expected distance of an edge selected in any round over the scenarios corresponding to different number $k$ of consumers available using (21). As the probability to select an edge does not change in the different scenarios we deduce in (18), (19) and (20) that, there is uniform probability of selecting any (consumer) edge given that producer $P_i$ is selected in a round. Intuitively this follows form symmetry as each consumer is equally likely to be picked in any round assuming equal capacities.

Using the law of total expectation over the producer selected in a round,

\begin{align}
E\Big[d_{i, j} \Big]   &= \sum_{i=1}^{m} E\Big[d_{i=P_i, j}  \bigm\vert i=P_i\Big] \cdot p\Big(i=P_i\Big) \\
&= \sum_{i=1}^{m} \bigg[  \frac{1}{n}  \cdot \sum_{j=1}^{n} d_{i,j} \bigg] \cdot \frac{1}{m} \\
&= \frac{1}{m \cdot n} \sum_{i=1}^{m} \sum_{j=1}^{n} d_{i,j} = \overline{d}
\label{operation 114}
\end{align}

So the expected cost of Algorithm ${\cal A}$ (\ref{alg2}) in any round is $E[G] = s \cdot E[d_{i,j}] = s \cdot \overline{d}$. The best possible cost of the optimal offline algorithm OPT in any round is when it chooses the cheapest edge $\min_{i, j} d_{i,j}$ in each round. Thus, the average-case (\cite{average-case}) competitive ratio is $\bigg( \frac{\overline{d}}{\min_{i, j} d_{i,j}} \bigg)$. The worst-case competitive ratio is $\bigg(\frac{d_{max}}{\min_{i, j} d_{i,j}}\bigg)$.

\end{proof}

\begin{theorem}  
Extending model ${\cal M}$ to ${\cal M}2$ with arbitrary producer request sizes, equal consumer capacities, arbitrary edge distances and allowing a single requests to be allocated across multiple consumers and running algorithm ${\cal A}2$, the expected total cost is $E[G] = \sum_{t=1}^{r} s(t) \cdot \overline{d}$.
\end{theorem}
\begin{proof} The simplified algorithm ${\cal A}2$ is:

\begin{algorithm}[H]
\caption{Online Algorithm ${\cal A}2$}
\label{alg3}
\begin{algorithmic}[1]
\STATE Given: $P$; $C$; $c$; $d_{i,j}$; $s$
\STATE Initialize: $\ell_j = 0$ $\forall j \in [n]$; $G = 0$; $n = |C|$
\FOR {$t = 1$ to $r$}
\STATE Choose $P_i \in [m]$ independently and uniformly at random
\STATE Choose an arbitrary request size $s(t) = \bigg(0, \sum_{j \in C} (c_j - \ell_j) \bigg]$
\FOR {$s=1$ to $s(t)$}
\STATE $C = \{c_j \: | \: j \in [n] \land \ell_j + s(t) \leq c\}$
\STATE Choose $c_j \in [n]$ independently and uniformly at random
\STATE $\ell_j = \ell_j + 1$
\STATE $G = G + d_{i,j}$
\ENDFOR
\ENDFOR
\end{algorithmic}
\end{algorithm}

By splitting the requests into unit sized blocks of size $f_D$ and using the fact that there is equal probability for the blocks to be allocated to any consumer, it follows from symmetry that a block of unit requests is equally likely to belong to any producer $P_i$. Using (20), (22) and (26) from {\bf Theorem 3} we get the expected cost of the edge selected in anay round as $\overline{d}$. So the expected cost is $E[G] = \sum_{t=1}^{r} s(t) \cdot \overline{d}$ and the average case competitive ratio is $\bigg( \frac{\overline{d}}{\min_{i, j} d_{i,j}} \bigg)$
\end{proof}

\begin{theorem}  
Extending model ${\cal M}2$ to ${\cal M}3$ with arbitrary producer demands, arbitrary consumer capacities and arbitrary edge distances and running algorithm ${\cal A}3$, the expected total cost is $E[G] = \sum_{t=1}^{r} s(t) \cdot \frac{\overline{d \cdot c}}{\overline{c}}, \; \forall t$.
\end{theorem}
\begin{proof} 
The simplified algorithm ${\cal A}3$ is,
\begin{algorithm}[H]
\caption{Online Algorithm ${\cal A}3$}
\label{alg4}
\begin{algorithmic}[1]
\STATE Given: $P$; $C$; $c$; $d_{i,j}$; $s$
\STATE Initialize: $\ell_j = 0$ $\forall j \in [n]$; $G = 0$; $n = |C|$
\FOR {$t = 1$ to $r$}
\STATE Choose $P_i \in [m]$ independently and uniformly at random
\STATE Choose an arbitrary request size $s(t) = \bigg(0, \sum_{j \in C} (c_j - \ell_j) \bigg]$
\FOR {$s=1$ to $s(t)$}
\STATE $C = \{c_j \: | \: j \in [n] \land \ell_j + s(t) \leq c\}$
\STATE Choose $c_j \in n$ with probability $\bigg( \frac{c_j}{\sum_{j=1}^{n} c_j } \bigg)$
\STATE $\ell_j = \ell_j + 1$
\STATE $G = G + d_{i,j}$
\ENDFOR
\ENDFOR
\end{algorithmic}
\end{algorithm}

Let $X_{i,j}$ be a indicator random variable that indicates that producer $P_i$ selects consumer $C_j$. As requests are split into unit sized blocks, the expected load on consumer $C_j$ assigned by producer $P_i$ is equal to value of $X_{i,j}$. Then $E[X_{i,j}] = E\Big[C_j \bigm\vert P_i\Big] = \frac{c_j}{\sum_{j \in n} c_j}$. Let $Y_j$ be the indicator random variable for the load placed on consumer $C_j$ in a specific round then,$E[Y_j] = \sum_{P_i=1}^{m} E[X_{P_i,j}] \cdot p(P_i) = \frac{c_j}{\sum_{j \in n} c_j}$ as each producer is equally likely to be selected in a round. Let $Y_j(t)$ denote the load placed on consumer $C_j$ after r requests have been completed then $E[Y_j(t)] = \sum_{t=1}^{r} s(t) \cdot E[Y_j] = \frac{c_j}{\sum_{j \in [n]} c_j} \cdot \sum_{t=1}^{r} s(t) $. Using unit sized requests the consumer load always remains proportional to its capacity.

As the producers are chosen uniformly at random the probability that a consumer $C_j$ was picked by a producer $P_i$ in any round is $p[P_i \bigm\vert C_j] = \frac{1}{m}$. Size of the request chosen $(0, \sum_{j \in C} (c_j - \ell_j)]$ at each iteration can be atmost equal to the remaining capacity available as it can be split amonst the consumers.

Expected distance of the edge when consumer $C_j$ is picked is,

\begin{align}
E\Big[d_{ij} \bigm\vert C_j\Big] &= \sum_{i=1}^{m} d_{i,j} \cdot p\Big(P_i \bigm\vert C_j\Big) \\
& = \frac{1}{m} \cdot \sum_{i=1}^{m} d_{i,j}
\label{operation 301}
\end{align}

The expected cost of edge picked in a any round assuming equal failure times,

\begin{align}
E[d_{i,j}] &= \sum_{j=1}^{n} E\Big[ d_{i,j} \bigm\vert c_{j} \Big] \cdot p(c_{j}) \\
&= \frac{1}{m} \cdot \sum_{j=1}^{n}  \sum_{i=1}^{m} d_{i,j} \cdot \frac{c_j}{\sum_{j=1}^{n} c_j } \\
&= \frac{n}{\sum_{j=1}^{n} c_j} \cdot \frac{\sum_{j=1}^{n}  \sum_{i=1}^{m} d_{i,j} \cdot c_j}{m \cdot n} = \frac{\overline{d \cdot c}}{\overline{c}}
\label{operation 302}
\end{align}

Expected cost is $E[G] = \sum_{t=1}^{r} s(t) \cdot \frac{\overline{d \cdot c}}{\overline{c}}$ and the average-case competitive ratio is $\bigg( \frac{\overline{d \cdot c}}{\overline{c} \cdot \min_{i,j} d_{i,j}} \bigg)$. Although the competitve ratio depends on the edge distances and consumer capacities, it is not possible to assume equal failure times using any other edge probability although this online algorithm $\mathcal{A}3$ is not optimal for the objective function in (\ref{operation 5}).

\end{proof}

\begin{figure} [h]
\centering
\includegraphics[width=120mm]{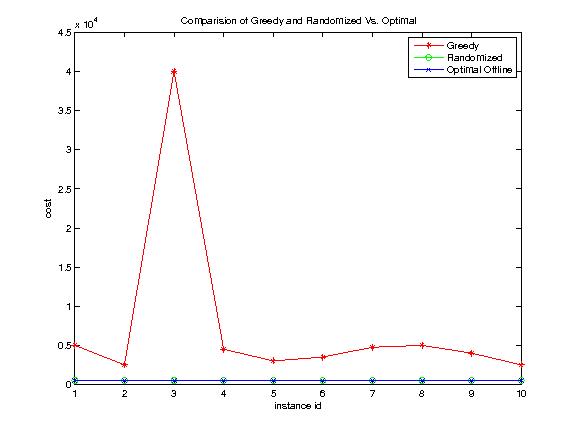}
\caption{Comparision of Greedy and Randomized Vs. Optimal Offline for special cases}
\label{comparision}
\end{figure}

The instances considered ranged from 1 to 100 producers and 1 to 100 consumers. The size of demands, edge distances and consumer capacities were picked randomly.

\section{Conclusion}
\label{sec:conclude}
The average-case competitive ratio of of $\bigg( \frac{\overline{d}}{\min_{i, j} d_{i,j}} \bigg)$ for the model with equal consumer capacities and arbitrary producer requests indicates that the performance depends on the quality of majority of links. For arbitrary consumer capacities with an average-case competitve ratio of $\bigg( \frac{\overline{d \cdot c}}{\overline{c} \cdot \min_{i,j} d_{i,j} } \bigg)$ the performance is decided by the quality of the links connected to the consumers with higher capacites. The optimal offline primal-dual algorithm runs in $O(r \cdot n)$ time whereas the online algorithms take $O(r)$ time where r and n are the number of requests and consumers respectively.

\end{document}